\documentclass[11pt,a4paper]{article}
\usepackage{amsmath}
\usepackage{amssymb}
\usepackage{amsthm}
\usepackage{psfrag}
\usepackage{epsfig}
\usepackage{wrapfig}

\usepackage{mathptmx}
\usepackage{fullpage}

\usepackage{ifpdf,xspace}

\usepackage{tikz}
\usetikzlibrary{positioning}
\usetikzlibrary{fit}

\newcommand{\ignore}[1]{}

\ignore{
\topmargin 0pt
\advance \topmargin by -\headheight
\advance \topmargin by -\headsep
\textheight 8.9in
\oddsidemargin 0pt
\evensidemargin \oddsidemargin
\marginparwidth 0.5in
\textwidth 6.5in
}

\newlength{\proofpostskipamount}
\newlength{\proofpreskipamount}
\newlength{\proofnegpostskip}
\newcommand{\REMOVE}[1]{}

\makeatletter
\newtheorem{rep@theorem}{\rep@title}{\bfseries}{\rmfamily}
\newcommand{\newreptheorem}[2]{%
\newenvironment{rep#1}[1]{%
 \def\rep@title{#2 \ref{##1}}%
 \begin{rep@theorem}}%
 {\end{rep@theorem}}}
\makeatother

\newtheorem{theorem}{Theorem}
\newtheorem{lemma}{Lemma}
\newtheorem{remark}{Remark}
\newcommand{\sset}[1]{\{ #1 \}}
\newcommand{\set}[2]{\{ #1 ; #2 \}}

\newcommand{\assign}{\mathbin{\raisebox{0.05ex}{\mbox{\rm :}}\!\!=}}

\newcommand{\PoA}{\mathit{PoA}}\newcommand{\poa}{\PoA} %\newcommand{\poa}{\mathrm{PoA}}
\newcommand{\npoa}{\mathit{ePoA}}\newcommand{\ePoA}{\mathit{ePoA}}

\newcommand{\CN}{C_N}
\newcommand{\Copt}{C_{\mathit{opt}}}
\newcommand{\hl}{\hat{\ell}}

\newcommand{\CM}{\hat{C}_N}
\newcommand{\CMN}{C_{\mathit{MN}}}
\newcommand{\CNM}{\CMN}

\date{}
\begin{document}

\title{Improving the Price of Anarchy for Selfish Routing via Coordination
Mechanisms}

\author{ 
	George Christodoulou\thanks{University of Liverpool, United Kingdom.
	\texttt{gchristo@liv.ac.uk}}
	\and
	Kurt Mehlhorn\thanks{ 
	Max-Planck-Institut f\"ur Informatik, Saarbr\"ucken, Germany.
	\texttt{mehlhorn@mpi-inf.mpg.de}}
	\and
	Evangelia Pyrga\thanks{Technische Universit\"at M\"unchen, Germany.
	\texttt{pyrga@in.tum.de}}
	}

\ignore{
\author{ 
	Giorgos Christodoulou\inst{1}
	\and Kurt Mehlhorn\inst{2}
	\and Evangelia Pyrga\inst{3} 
	}

\institute{University of Liverpool, United Kingdom.
        \\ \email{gchristo@liv.ac.uk}\and 
Max-Planck-Institut f\"ur Informatik, Saarbr\"ucken, Germany.
       \\ \email{mehlhorn@mpi-inf.mpg.de}
\and Technische Universit\"at M\"unchen, Germany.
       \\ \email{pyrga@in.tum.de} 
}
}

\maketitle

\begin{abstract} 
We reconsider the well-studied Selfish Routing game with 
affine latency functions. The Price of Anarchy for this class of 
games takes maximum value 4/3; this maximum is 
attained already  for a simple 
network of two parallel links, known as Pigou's network. We improve upon 
the value 4/3 by means of Coordination Mechanisms. 
 
We increase the latency functions of the edges in the network, i.e., if 
$\ell_e(x)$ is the latency function of an edge $e$, we replace it by 
$\hat{\ell}_e(x)$ with $\ell_e(x) \le \hat{\ell}_e(x)$ for all $x$. 
Then an adversary fixes a 
demand rate as input. The \emph{engineered Price of Anarchy} of the 
mechanism is defined as the worst-case ratio of the Nash social cost in the modified 
network over the optimal social cost in the original network. Formally, if 
$\CM(r)$ denotes the cost of the worst Nash flow in the modified network for 
rate $r$ and $\Copt(r)$ denotes the cost of the optimal flow in the original 
network for the same rate then 
\[   \ePoA   = \max_{r \ge 0} \frac{\CM(r)}{\Copt(r)}.\] 
  
We first exhibit a simple coordination mechanism that achieves for any network of parallel 
links an engineered Price of Anarchy strictly less than 
4/3. For the case of two parallel links 
our basic mechanism gives 5/4 = 1.25. Then, for the case of two parallel links, we 
describe an {\em optimal} mechanism; its engineered Price of Anarchy lies between 1.191 
and 1.192. 
% \keywords{Algorithmic Game Theory \and  Selfish Routing  \and Price of Anarchy   
% \and Coordination Mechanisms } 
% \PACS{PACS code1 \and PACS code2 \and more} 
% \subclass{MSC code1 \and MSC code2 \and more} 
\end{abstract}

\section{Introduction} 
\label{sec: introduction}

We reconsider the well-studied Selfish Routing game with affine cost
functions and ask whether increasing the cost functions can reduce the cost of
a Nash flow. In other words, the increased cost functions should induce a user
behavior that reduces cost despite the fact that the cost is now determined by increased
cost functions. We answer the question positively in the following sense. The
Price of Anarchy, defined as the maximum ratio of Nash cost to optimal cost, is
4/3 for this class of games. We show that increasing costs can reduce the price
of anarchy to a value strictly below 4/3 at least for the case of networks of
parallel links. For a network of two parallel links, we reduce the price of
anarchy to a value between 1.191 and 1.192 and prove that this is optimal. In
order to state our results precisely, we need some definitions. 

We consider single-commodity congestion
games on networks, defined by a directed graph $G=(V,E)$, designated nodes
$s,t\in V$, 
and a set $\ell=(\ell_e)_{e\in E}$ of non-decreasing, non-negative functions; $\ell_e$ 
is the latency function of edge $e\in E$. 
Let $P$ be the set of all paths from $s$ to $t$, and let  $f(r)$ be  a feasible
$s,t$-flow routing $r$ units of flow.  For any $p\in P$, let $f_p(r)$ denote the
amount of flow that $f(r)$ routes via path $p$. 
For ease of notation, when $r$ is fixed and clear from context, we will write simply $f_, f_p$ instead of $f(r), f_p(r)$.
By definition, $\sum_{p\in P} f_p
= r$. Similarly, for any edge $e\in E$, let $f_e$ be the amount of flow going
through $e$. We define the 
latency of $p$ under flow $f$  as $\ell_p(f) = \sum_{e\in
p}\ell_e(f_e)$ and the cost of flow $f$ as $C(f) = \sum_{e\in E} f_e \cdot
\ell_e(f_e)$ and use $\Copt(r)$ to denote the minimum cost of any flow of rate $r$. We will refer to such a minimum cost flow as an  \emph{optimal} flow (Opt).
A feasible flow $f$ that routes $r$ units of flow from $s$ to $t$ 
is at \emph{Nash (or Wardrop~\cite{War52}) Equilibrium\footnote{This assumes
continuity and monotonicity of the latency functions. For
non-continuous functions, see the
discussion later in this section.}} if for $p_1, p_2 \in P$ with $f_{p_1}>0$, 
$\ell_{p_1}(f)\le \ell_{p_2}(f)$. We use $\CN(r)$ to denote the maximum cost of a Nash
%\kmnote{I moved the complicated definition to the section on lower bounds}
\ignore{

We say that a feasible flow $f$ that routes $r$ units of flow from $s$ to $t$
is at \emph{Nash Equilibrium} if for $p_1, p_2 \in P$,  with $f_{p_1}>0$ and
$\delta\to 0$,  $\ell_{p_1}(f)\le \ell_{p_2}(f')$, where
\begin{equation}
\label{eq:ds-definition}
        f_p' = \left\{
                \begin{array} {ll}
                f_p - \delta, &         p = p_1 \\
                f_p + \delta, &         p = p_2 \\
                f_p             &       p \neq p_1,p_2. \\
                \end{array}
        \right.
\end{equation}
\epnote{I used our old definition from Dafermos-Sparrow, but changed $\delta$ to
being infinitesimally small. }
%A feasible flow $f$ that routes $r$ units of flow from $s$ to $t$ 
%is at \emph{Nash (or Wardrop~\cite{War52}) Equilibrium} \marginpar{we need general definition} if for $p_1, p_2 \in P$ with $f_{p_1}>0$, 
%$\ell_{p_1}(f)\le \ell_{p_2}(f)$. 
We use $\CN(r)$ to denote the maximum cost of a Nash}%
flow for rate $r$. The \emph{Price of Anarchy (PoA)}~\cite{KP09} (for demand $r$) is defined as 
\[ \PoA(r) =  \frac {\CN(r)} {\Copt(r)} \quad\text{and}\quad \PoA = \max_{r > 
0} \PoA(r).\]
PoA is bounded by
$4/3$ in the case of affine latency functions $\ell_e(x) = a_ex + b_e$ with $a_e \ge
0$ and $b_e \ge 0$; see~\cite{Roughgarden-Tardos,Correa-et-al}. The worst-case is already assumed for a simple network of two
parallel links, known as Pigou's network; see Figure~\ref{fig:Pigou}. 

%shown in Figure~\ref{fig: Pigou}

\newcommand{\hell}{\hat{\ell}}

% \begin{figure}{0.4\textwidth}
% \huge{TODO}
% \caption{Pigou's network}\label{fig:Pigou}\end{figure}
%\begin{figure}[t]
\begin{figure}
\begin{center}
{\small
\psfrag{l1}{$\ell_1(x) = x$}\psfrag{l2}{$\ell_2(x) = 1$}\psfrag{f1*r}{$f_1^*(r) = \min\{1/2,r\}$}\psfrag{f2*r}{$f_2^*(r) =
\max\{0,r - 1/2\}$}\psfrag{f1Nr}{$f_1^N(r) = \min\{1,r\}$}\psfrag{f2Nr}{$f_2^N(r) =
\max\{0,r - 1\}$}\psfrag{h2}{$\hell_2(x) = 1$} \psfrag{h1}{$\hell_1(x) =
\begin{cases} x & \text{for $x \le 1/2$}\\ \infty & \text{for $x > 1/2$}
\end{cases}$}
\psfrag{r}{$r$}\psfrag{PoA}{$\PoA(r)$}
	\epsfig{file=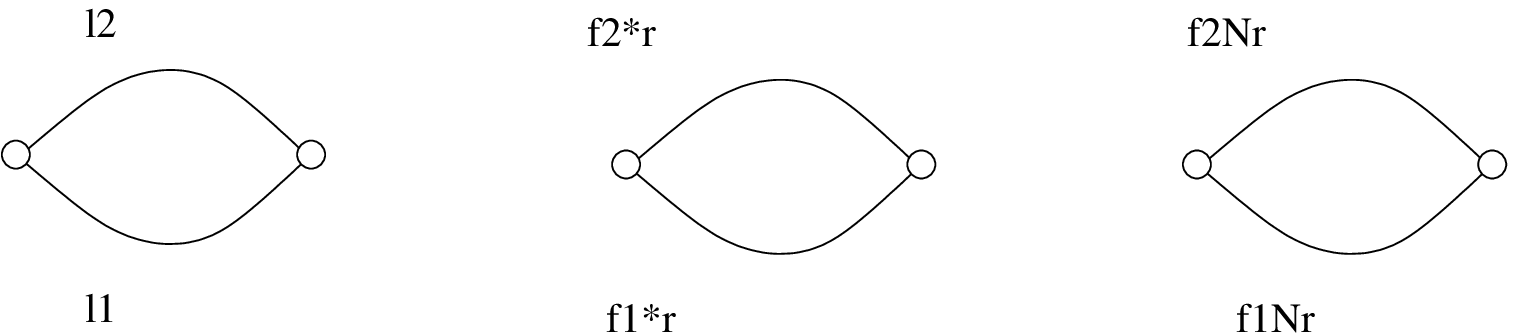, height=3cm}}

\centering{ (a)  \hspace{5cm} (b) \hspace{5cm} (c)}

\vspace{1ex}
\begin{tikzpicture}[auto, right,  node distance=0.5cm, scale=.37]
	%\tikzstyle{vertex}=[shape=circle,draw=black,line width=0.3ex, inner sep =0.3cm]
        %\tikzstyle{vertex2}=[vertex,draw=blue, line width=0.4ex]
	\tikzstyle{axis}=[->,draw=black,rounded corners]
	\tikzstyle{graph}=[draw=black]
	\tikzstyle{graph-blue}=[dotted]
	
	\node (ph1) at (-1.5cm, -1cm) {};
	\node (ph2) at (7.5cm, 5.0cm) {};

	\coordinate(origin) at (0,0) {};
	\coordinate(xmax) at (7.5cm,0) {};
	\coordinate(ymax) at (0,6cm) {};

	\draw[axis] (origin.center) -- (xmax); 
	\draw[axis] (origin.center) -- (ymax);

	\node[anchor=east] (ylabel) at ([xshift=-1ex, yshift=0ex]ymax.east) {$\poa(r)$};
	%\node[anchor=south] (ylabel) at (ymax.north) {$\poa(r)$};
	\node[anchor=north] (xlabel) at ([xshift=-2ex, yshift=-1ex]xmax.south) {$r$};

	\coordinate(ropt) at (1,1) {};

	%\draw[graph] (0,1)   --  (2,1) node[midway] {x} 
	\draw[graph] (0,1) coordinate (graphbegin)  --  (ropt);  
	\draw[graph-blue] (1,0) coordinate (r0) --(ropt);	
	\node[anchor=north]  at ([yshift=-1ex]r0.south) {$\frac 1 2$};
	\node[anchor=east]  at ([xshift=-1ex]graphbegin.east) {1};

	\draw[graph] (ropt)  
			.. controls (2.1,1.3) and (2.35,1.6)..
			(2.5cm,4) coordinate (maxPoA); 
	\draw[graph-blue] (0,4) coordinate (XmaxP) --(maxPoA);	
	\draw[graph-blue] (2.5cm,0) coordinate (Px) --(maxPoA);	

	\node[anchor=east]  at ([xshift=-1ex]XmaxP.east) {4/3};
	\node[anchor=north]  at ([yshift=-1ex]Px.south) {$ 1$};

	\draw[graph] (maxPoA)  
			.. controls (3,1.6cm) and (3.9cm,1.0cm)..
			 (7,1.0cm) coordinate (limit); 

	\end{tikzpicture}
{\small
\psfrag{l1}{$\ell_1(x) = x$}\psfrag{l2}{$\ell_2(x) = 1$}\psfrag{f1*r}{$f_1^*(r) = \min\{1/2,r\}$}\psfrag{f2*r}{$f_2^*(r) =
\max\{0,r - 1/2\}$}\psfrag{f1Nr}{$f_1^N(r) = \min\{1,r\}$}\psfrag{f2Nr}{$f_2^N(r) =
\max\{0,r - 1\}$}\psfrag{h2}{$\hell_2(x) = 1$} \psfrag{h1}{$\hell_1(x) =
\begin{cases} x & \text{for $x \le 1/2$}\\ \infty & \text{for $x > 1/2$}
\end{cases}$}
\psfrag{r}{$r$}\psfrag{PoA}{$\PoA(r)$}
	\epsfig{file=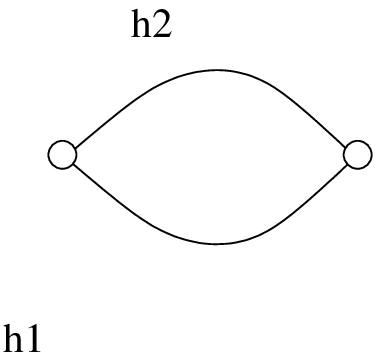, height=3cm}}
\medskip

(d) \hspace{4cm} (e)
\end{center}

\caption{Pigou's network: We show the original network in (a), the optimal flow
  in (b) and
the Nash flow in (c) as a function of the rate $r$, respectively. The Price of Anarchy
as a function of the rate is shown in (d); $PoA(r)$ is 1 for $r \le 1/2$, then starts to grow
until it reaches its maximum of $4/3$ at $r = 1$, and then decreases again and
approaches $1$ as $r$ goes to infinity. Finally, in (e) we show the 
modified latency functions. %as changed by our mechanism. 
We obtain $\ePoA(r) = 1$ for all
$r$ in the case of Pigou's network.
\label{fig:Pigou}}
\end{figure}

A {\em Coordination Mechanism}\footnote{Technically, we consider {\em
    symmetric} coordination mechanisms in this work, as defined in
  \cite{CKN09} i.e., the latency modifications affect the users in a
  symmetric fashion.} replaces the cost functions $(\ell_e)_{e \in E}$ by
functions\footnote{One can interpret the difference $\hell_e-\ell_e$ as a
  flow-dependent toll imposed on the edge $e$.} $\hell = (\hell_e)_{e\in E}$ such that $\hell_e(x)\ge
\ell_e(x)$ for all $x\ge 0$.  Let $\hat C(f)$ be the cost of flow $f$
when for each edge $e\in E$, $\hat \ell_e$ is  used instead of $\ell_e$ and let
$\CM(r)$ be the maximum cost of a Nash flow of rate $r$ for the modified latency
functions. We define the \emph{engineered Price
of Anarchy} (for demand $r$) as 
\[  \npoa(r) = \frac {\CM(r)} {\Copt(r)}\quad\text{and}\quad \npoa = \max_{r > 0} \npoa(r).\]
We stress that the optimal cost refers to the original latency
functions $\ell$. 
\\
\\
%\paragraph{Non-continuous latency functions:}
\noindent {\bf Non-continuous Latency Functions:}
In the previous definition, as it will become clear in
Section~\ref{sec: negative result}, it is important to allow non-continuous
modified latencies. However, when we move from continuous to non-continuous latency functions, Wardrop
equilibria do not always exist. Non-continuous functions
have been studied by transport economists to model the effects of 
step-function congestion tolls and traffic lights. Several notions of equilibrium that
handle discontinuities have been proposed in the
literature\footnote{See~\cite{Pat94,MP07} for an excellent exposure of the
  relevant concepts, the relation among them, as well as for
  conditions that guarantee their existence.}. The ones
that are closer in spirit to Nash equilibria, are those proposed
by Dafermos\footnote{In~\cite{Daf71}, Dafermos weakened
  the orginal definition by \cite{DS69} to make it closer
  to the concept of Nash Equilibrium.}~\cite{Daf71} and Berstein and Smith~\cite{BS94}.  
According to the Dafermos'~\cite{Daf71} definition of {\em user optimization}, a flow is in equilibrium if no
{\em sufficiently small} fraction of the users on any path, can decrease the latency
they experience by switching to another path\footnote{See
  Section~\ref{lower bound} for a formal definition.}. Berstein and Smith~\cite{BS94}
introduced the concept of {\em User Equilibrium}, weakening
further the Dafermos equilibrium, taking the fraction of the users to the
limit approaching 0. The main idea of their definition is to capture the notion of the
{\em individual commuter}, that is implicit in Wardrop's
definition for continuous functions. The Dafermos equilibrium on the
other hand is a stronger concept that captures the notion of coordinated deviations by {\em groups of commuters}.

We adopt the concept of User Equilibrium. Formally, we say that
a feasible flow $f$ that routes $r$ units of flow from $s$ to $t$ is a
User Equilibrium, iff for all $p_1, p_2 \in P$ with $f_{p_1}>0$, 

\begin{equation}\ell_{p_1}(f)\leq \lim\inf_{\epsilon\downarrow 0}\ell_{p_2}(f+\epsilon
\mathbf{1}_{p_2} -\epsilon \mathbf{1}_{p_1}), \label{user equilibrium} \end{equation}
where $\mathbf{1}_p$ denotes the flow where only one unit passes along
a path $p$.

Note that for continuous functions the above definition is identical
to the Wardrop Equilibrium. One has to be careful when designing a
Coordination Mechanism with discontinuous functions, because the
existence of equilibria is not always guaranteed\footnote{See for
  example \cite{PN98,BS94} for examples where equilibria do not exist even for the
  simplest case of two parallel links and non-decreasing functions.}.
It is important to emphasize, that all the mechanisms that we suggest in this paper use both
lower semicontinuous and regular\footnote{See~\cite{BS94} for a
  definition of regular functions.} latencies, and therefore User Equilibrium existence is
guaranteed due to the theorem of~\cite{BS94}. Moreover, since our
modified latencies are non-decreasing, all User Equilibria are also
Dafermos-Sparrow equilibria. 
%Finally, the lower bounds that we provide, do not rely on any monotonicity
%assumptions, and hold for general coordination mechanisms as defined above. 
From now on, we refer to the
User Equilibria as Nash Equilibria, or simply Nash flows.
\\ 
\\
%\paragraph{Our Contribution:} 
\noindent {\bf Our Contribution:}
We demonstrate the possibility of
reducing the Price of Anarchy for Selfish Routing via Coordination Mechanisms. 
We obtain the following results for networks of $k$ parallel links. 
\begin{itemize}
\item if original and modified latency functions are continuous, no improvement is
possible, i.e., $\npoa \ge \poa$; see Section~\ref{sec: negative result}. 
\item for the case of affine cost functions, we describe a simple 
coordination mechanism that achieves an engineered Price of Anarchy strictly less than
4/3; see Section~\ref{sec: mechanism}. The functions $\hell_e$ are of the form
\begin{equation}\label{eq: simple}   \hell_e(x) = \begin{cases}  \ell_e(x)    & \text{for $x \le r_e$}\\
                                  \infty       & \text{for $x > r_e$}.
\end{cases} \end{equation}
For the case of two parallel links, the mechanism gives 5/4 (see Section~\ref{sec: two links}), for
Pigou's network it gives 1, see Figure~\ref{fig:Pigou}.
\item For the case of two parallel links with affine cost functions, we
describe an {\em optimal
\iffalse
\footnote{The lower bound that we provide in
    Section~\ref{lower bound} holds for all deterministic coordination
    mechanisms with respect to the Dafermos-Sparrow\cite{DS69}
    definition of equilibria. However, the arguments of our proof work for
     all deterministic coordination mechanisms that use {\em non-decreasing
    modified latencies} even for the weaker definition of User
  Equilibrium. Therefore the mechanism of Section~\ref{sec: two links
    advanced} is optimal for these two classes of mechanisms.}
\fi
\footnote{The lower bound that we provide in
    Section~\ref{lower bound} holds for all deterministic coordination
    mechanisms that use {\em non-decreasing
    modified latencies}, with respect to both notions of equilibrium described
    in the previous paragraph.}%
} mechanism; its engineered Price of Anarchy lies between 1.191
and 1.192 (see Sections~\ref{sec: two links advanced} and~\ref{lower
bound}). It uses modified cost functions of the form  
\begin{equation}\label{eq: refined}  \hell_e(x) = \begin{cases}  \ell_e(x)    & \text{for $x \le r_e$ and $x
\ge u_e$}\\
                                  \ell_e(u_e)       & \text{for $r_e < x < u_e$}.
\end{cases} \end{equation}
\end{itemize}

\noindent
The Price of Anarchy is a standard measure to quantify the effect of
selfish behavior. There is a vast literature studying the Price of
Anarchy for various models of selfish routing and scheduling problems
(see~\cite{NRTV07}). We show that simple coordination mechanisms can reduce the Price of Anarchy for
selfish routing games below the 4/3 worst case for networks of parallel links
and affine cost functions.

We believe that our arguments extend to more general
cost functions, e.g., polynomial cost functions. However, the restriction to
parallel links is crucial for our proof. We leave it as a major open problem to
prove results for general networks or at least more general networks, e.g.,
series-parallel networks. 
\\
\\
%\paragraph{Implementation:} 
\noindent {\bf Implementation:} We discuss the realization of the modified cost
function in a simple traffic scenario where the driving speed on a link is
a decreasing function of the flow on the link and hence the transit time is an
increasing function. The step function in (\ref{eq: refined})
can be realized by setting a speed limit corresponding to transit
time $\ell_e(u_e)$ once the flow is above $r_e$. The functions in (\ref{eq:
simple}) can be approximately realized by access control. In any time unit only
$r_e$ items are allowed to enter the link. If the usage rate of the link is
above $r_e$, the queue in front of the link will grow indefinitely and hence
transit time will go to infinity. 
\\
%\paragraph{Related Work:}
% \paragraph{Related Work:} The difference $\hell_e(x) - \ell_e(x)$ may be
% interpreted as a toll. The effect of tolls on Nash equilibria is
% well-studied~\cite{Cole-et-al,Fleischer05,Bonifaci}. It is well known that
% so-called marginal cost tolls, i.e., $\hell_e(x) = \ell_e(x) + 
% x \ell'_e(x)$, result in a Nash flow that it equal to the optimum flow for the
% original cost functions. 

% Cole et
% al.~\cite{Cole-et-al,Roughgarden} study the question whether tolls can reduce the cost of a Nash
% equilibrium. They show:
% \begin{quote}
% \begin{itemize}
% \item In every network with linear latency functions, marginal cost taxes do not improve the cost of a flow at 
% Nash equilibrium. 
% \item The maximum-possible benefit of taxes is no more than that of edge removals. 
% \item For every 
% network with linear latency functions -- not merely worst-case examples -- taxes cannot decrease 
% the cost of a flow at Nash equilibrium beyond what can be achieved by removing
% edges. 
% \end{itemize}
% \end{quote} 
% For these mainly negative results it is assumed that the modified cost
% functions are continuous
% and monotone. Our modified cost functions are monotone, but not continuous. 
% Our simple strategy can be viewed as a
% generalization of link removal. Removal of a link reduces the capacity of the
% edge to zero, our simple strategy reduces the capacity to a threshold $r_e$. 

%\paragraph{Related Work:} 

\noindent {\bf Related Work:}
The concept of Coordination Mechanisms was introduced in (the conference
version of)~\cite{CKN09}. 
%They proposed an algorithmic
%framework to minimize the Price of Anarchy for Congestion Games by appropriately
%modifying the latency functions on the edges i.e., by introducing
%delays and by changing the priority policies on the edges. 
Coordination Mechanisms have been
%mainly 
used to improve the Price of Anarchy in scheduling problems for
parallel and related machines~\cite{CKN09,ILMS05,Kol08} 
as well as for unrelated
machines~\cite{AJM07,Caragiannis09}; the objective is
makespan minimization. Very recently, \cite{CCG+11} considered as an
objective the weighted sum of completion times. Truthful coordination mechanisms
have been studied in \cite{ABP06,CGP07,AngelBPT09}. 

Another very well-studied attempt to cope with selfish behavior is the
introduction of taxes (tolls) on the edges of the network in selfish
routing games~\cite{CDR03,FJM04,KK04a,KK04b,Fleischer05,Bonifaci}. 
The disutility of a player is modified and equals her
latency plus some toll for every edge that is used in her path. 
It is well known (see for example \cite{CDR03,FJM04,KK04a,KK04b}) that
so-called marginal cost tolls, i.e., $\hell_e(x) = \ell_e(x) + 
x \ell'_e(x)$, result in a Nash flow that is equal to the optimum flow for the
original cost functions.\footnote{It is important to observe that although the Nash
flow is equal to the optimum flow, its cost with respect to the marginal cost
function can be twice as large as its cost with respect to the original cost
function. For Pigou's network, the marginal costs are $\hell_1(x) = 2x$ and
$\hell_2(x) = 1$. The cost of a Nash flow of rate $r$ with $r \le 1/2$ is
$2r^2$ with respect to marginal costs; the cost of the same flow with respect
to the original cost functions is $r^2$.}
Roughgarden~\cite{R01} seeks a subnetwork of a given
network that has optimal Price of Anarchy for a given demand. 
\cite{Cole-et-al} studies the question whether tolls can reduce the cost of a Nash
equilibrium. They show that for networks with affine latencies,
marginal cost pricing does not improve the cost of a flow at Nash
equilibrium, as well as that the maximum possible benefit that one can
get is no more than that of edge removal. 
\\
\\
%\paragraph{Discussion:}
\noindent {\bf Discussion:}
The results of this paper are similar in spirit to the results discussed in the
previous paragraph, but also very different. The above papers assume that taxes
or tolls are determined with full knowledge of the demand rate $r$. Our
coordination mechanisms must {\em a priori} decide on the modified
latency functions \emph{without knowledge of the demand}; it must 
determine the modified functions
$\hell$ and then an adversary selects the input rate $r$. More
importantly, our target objectives are different; we want to minimize
the ratio of the modified cost (taking into account the increase of
the latencies) over the {\em original} optimal cost.%
\ignore{The negative results of \cite{Cole-et-al} hold for marginal cost prices that 
correspond to continuous and monotone modified cost functions. 
Therefore, one cannot expect to
achieve a Price of Anarchy strictly less than 4/3 by use of marginal cost
pricing, even if the demand rate is known. In fact, in Section~\ref{sec: negative result} we show that
for parallel links one cannot improve the Price of Anarchy even with general
continuous monotone functions.
Our modified cost functions are
monotone, but not continuous.} Our simple strategy presented in
Section~\ref{sec: mechanism} can be viewed as a generalization of link removal.
Removal of a link reduces the capacity of the edge to zero, our simple strategy
reduces the capacity to a threshold
$r_e$.  Following~\cite{CKN09}, we study \emph{local} mechanisms; the decision of modifying
  the latency of a link is taken based on the amount of flow that
  comes through the particular link \emph{only}.\footnote{It is not hard to see that, similarly to the case where the demand is
known, using global flow information (at least for the case of parallel
links) can lead to mechanisms with $\ePoA=1$. We would like to thank
Nicol\'as Stier Moses for making us emphasizing that distinction.}

\section{Continuous Latency Functions Yield No Improvement}\label{sec: negative result}

The network in this section consists of $k$ parallel links connecting $s$ to
$t$ and the original latency functions are assumed to be continuous and
non-decreasing. 
We show 
that substituting them by continuous functions brings no improvement.

\begin{lemma}\label{lem:continuous-lb} Assume that the original functions $\ell_e$ are continuous and non-decreasing. 
	Consider some modified latency functions $\hat \ell$ and some rate $r$ for which
	there is a Nash Equilibrium flow $\hat f$ such that the latency function $\hat \ell_i$ is 
	continuous at $\hat f_i(r)$ 
	for all $1\le i\le k$. Then $\npoa(r)\ge \poa(r)$.
\end{lemma}

\begin{proof} 
It is enough to show that 
	$ \CM(r)\ge C_N(r)$.
Let $f$ be a Nash flow for rate $r$ and the original cost functions. If $f =
\hat f$, the claim is obvious. 
If $\hat f \not= f$, there must be a $j$ with $\hat f_j(r) > f_j(r)$.
The local continuity of $\hat \ell_i$ at  $\hat f_i(r)$,  implies that
	$\hat \ell_i(\hat f_i(r) ) = \hat \ell_{i'}(\hat f_{i'}(r)) $, for all 
	$i,i'\le k$  such that $\hat f_i(r), \hat f_{i'}(r) >0$.
Therefore, 
\[ \CM(r) =	\hat C(\hat f(r))= \sum_{i=1}^{k} \hat f_i(r) \hat \ell_i(\hat f_i(r))
		= r\cdot \hat \ell_j(\hat f_j(r)) 
		\ge 	r\cdot \ell_j(\hat f_j(r)) 		
		\ge 	r\cdot \ell_j( f_j(r)) \]
since $\hell_j(x) \ge \ell_j(x)$ for all $x$ and $\ell_j$ is non-decreasing. Since
$f$ is a Nash flow we have $\ell_{i}(f_{i}(r))\le \ell_j(f_j(r))$ for any $i$
with $f_{i}(r)>0$. Thus
\[ C_N(r) = \sum_{i=1}^{k}  f_i(r) \ell_i( f_i(r)) \le r \cdot \ell_{j}(
f_{j}(r)). \]
 \end{proof}

\section{A Simple Coordination Mechanism}\label{sec: mechanism}

Let $\ell_i(x) = a_i x + b_i = (x + \gamma_i)/\lambda_i$ be the latency
function of the $i$-th link, $1\le i \le k$. We call $\lambda_i$ the {\em efficiency}
of the link. We order the links in order of increasing $b$-value and
assume $0\leq b_1 < b_2 < \ldots <
b_{k}$ as two links with the same $b$-value may be combined (by adding their
efficiencies). We say that a link is {\em used} if it
carries positive flow. We may assume $a_i > 0$ for all $i < k$; if $a_i = 0$,
links $i+1$ and higher will never be used.  The following theorem summarizes some basic facts about
optimal flows and Nash flows; it is proved by straightforward
calculations.\footnote{In a Nash flow all used links have the same
latency. Thus, if $j$ links are used at rate $r$ and $f_i^N$ is the flow on the
$i$-th link, then $a_1 f_1^N + b_1 = \ldots = a_j f_j^N + b_j \le b_{j+1}$ and
$r = f_1^N + \ldots + f_j^N$. The values for $r_j$ and $f_i^N$ follow from
this. Similarly, in an optimal flow all used links have
the same marginal costs.} We state the theorem for the case that $a_k$ is
positive. The theorem is readily extended to the case $a_k = 0$ by letting
$a_k$ go to zero and determining the limit values. We will only use the theorem in
situations, where $a_k > 0$.

\begin{theorem}\label{thm: basic facts} Let 
%$\gamma_1 > \gamma_2 > \ldots > \gamma_k$ 
$0\leq b_1 <b_2 < \ldots < b_k$ 
and $\lambda_i \geq 0$ for all $i$. Let 
$\Lambda_j = \sum_{i \le j} \lambda_i$ and $\Gamma_j =
\sum_{i \le j} \gamma_j$. Consider a fixed rate $r$ and let $f_i^*$ and $f_i^N$, $1 \le i
\le k$, be the optimal flow and the Nash flow for rate $r$
respectively. Let 
\[ r_j = \sum_{1 \le i < j} (b_{j} - b_i)\lambda_i = \sum_{1 \le i < j} (b_{i+1} - b_i)\Lambda_i. \]
Then

\noindent (a) $\Gamma_j + r_j = b_j \Lambda_j$ and $\Gamma_{j-1} + r_j =
b_j \Lambda_{j-1}$. 
\\

\noindent (b) If Nash uses exactly $j$ links at rate $r$ then 
\[ r_j\leq r \leq r_{j+1},\quad\quad
   f_i^N =  \frac{r \lambda_i}{\Lambda_j} + \delta_i, \quad\text{where }
\delta_i = \frac{\Gamma_j \lambda_i}{\Lambda_j} - \gamma_i,\quad\text{and}\quad
\CN(r) = \frac{1}{\Lambda_j} \left( r^2 + \Gamma_j r \right).\]

\noindent (c) 
If Opt uses exactly $j$ links at rate $r$ then 
\[\frac{r_j}{2}\leq r \leq \frac{r_{j+1}}{2},\quad
f_i^* =  \frac{r \lambda_i}{\Lambda_j} + \delta_i/2, \quad\text{where }
\delta_i = \frac{\Gamma_j \lambda_i}{\Lambda_j} - \gamma_i, \]
and 
\[ \Copt(r) = \frac{1}{\Lambda_j} \left( r^2 + \Gamma_j r \right) - \sum_{i \le
j} \frac{\delta_i^2}{4\lambda_i} =  \frac{1}{\Lambda_j} \left( r^2 + \Gamma_j r\right)
- C_j , \text{where } C_j = \left(\sum_{i=1}^{h}\sum_{h=i}^{j}(b_h - b_i)^2
\lambda_h \lambda_i\right)/(4\Lambda_j) .\]

\noindent (d) If $s < r$ and Opt uses exactly $j$ links at $s$ and $r$ then 
\[  \Copt(r) = \Copt(s) + \frac{1}{\Lambda_j} \left( (r - s)^2 + (\Gamma_j +
2s) (r -s) \right).\]

\noindent (e) If $s < r$ and Nash uses exactly $j$ links at $s$ and $r$ then 
\[  \CN(r) = \CN(s) + \frac{1}{\Lambda_j} \left( (r - s)^2 + (\Gamma_j +
2s) (r -s) \right).\]

% \Kurt{\noindent (f) Finally, $\Gamma_j + r_j = b_j \Lambda_j$ and $\Gamma_{j-1} + r_j =
% b_j \Lambda_{j-1}$. I removed (f) because it is identical to (a)}{}

\end{theorem}

% \begin{proof}
% \noindent (1) is trivial.

% \noindent (2) In a Nash flow all used links have the same latency $L$. Thus, if the
% first $j$ links are used at rate $r$ and $f_i^N$ is the flow on the
% $i$-th used link, then $$L=\frac{f_1^N + \gamma_1}{\lambda_1} = \ldots = \frac{f_j^N + \gamma_j}{\lambda_j} \le b_{j+1}$$ and
% $r = f_1^N + \ldots + f_j^N$. 

% Therefore $L=\frac{r+\Gamma_j}{\Lambda_j}$, or $r=L\cdot
% \Lambda_j-\Gamma_j$, and by use of (1) it's easy to see $r_j\leq r \leq r_{j+1}$.

% $$f_i^N=\lambda_i\cdot L-\gamma_i=\lambda_i\cdot
% \frac{r+\Gamma_j}{\Lambda_j} - \gamma_i.$$

% \noindent (3) 

% The values for $r_j$ and $f_i^N$ follow from
% this. Similarly, in an optimal flow all used links have
% the same marginal costs.

% \noindent (4) 
% \noindent (5)
% \noindent (6)   

% \begin{eqnarray*}
% r
% \end{eqnarray*}

% \end{proof}

We next define our simple coordination mechanism. In the case of $k$
  links, it is governed by parameters $R_1$, $R_2, \dots, R_{k-1}$; $R_i \ge 2$ for all $i$. We call the $j$-th link
\emph{super-efficient} (with respect to parameters $R_1$ to $R_{k-1}$)
if $\lambda_{j} > R_{j-1} \Lambda_{j-1}$. In Pigou's
network (see Figure~\ref{fig:Pigou}), the second link is super-efficient for any choice of $R_1$ since
$\lambda_2 = \infty$ and $\Lambda_1 = \lambda_1 = 1$. Super-efficient links are the cause
of high Price of Anarchy. Observe that Opt starts using the $j$-th link at rate
$r_j/2$ and Nash starts using it at rate $r_j$. If the $j$-th link is
super-efficient, Opt will send a significant fraction of the total flow across
the $j$-th link and this will result in a high Price of Anarchy. Our
coordination mechanism induces the Nash flow to use super-efficient links
earlier. The latency functions $\hell_i$ are
defined as follows: $\hell_i = \ell_i$ if there is no super-efficient link $j > i$; in particular the latency function of the highest
link (= link $k$) is unchanged. Otherwise, we choose a threshold value $T_i$
(see below) and set $\hell_i(x) = \ell_i(x)$ for $x \le
T_i$ and $\hell_i(x) = \infty$ for $x > T_i$. The threshold values are chosen
so that the following behavior results. We call this behavior \emph{modified
Nash (MN)}.

Assume that Opt uses $h$ links, i.e., $r_h/2 \le r \le r_{h+1}/2$. If
$\lambda_{i+1} \le R_i \Lambda_i$ for all $i$, $1 \le i < h$, MN behaves like
Nash. Otherwise, let $j$ be minimal such that link $j+1$ is super-efficient; 
MN changes its behavior at rate $r_{j+1}/2$. More precisely, it freezes the flow across
the first $j$ links at their current values when the total flow is equal to $r_{j+1}/2$
and routes any additional flow across links
$j+1$ to $k$. The thresholds for the lower links are chosen in such a way that this
freezing effect takes place. The additional flow is routed by using the
strategy recursively. In other
words, let $j_1 + 1$, \ldots, $j_t + 1$ be the indices of the super-efficient
links. Then MN changes behavior at rates $r_{j_i + 1}/2$. At this rate the
flow across links $1$ to $j_i$ is frozen and additional flow is routed across
the higher links. 

We use $\CMN(r) = \CM^{R_1,\ldots,R_{k - 1}}(r)$ to denote the cost of MN at
rate $r$ when operated with parameters $R_1$ to $R_{k - 1}$. Then $\ePoA(r) =
\CMN(r)/\Copt(r)$. For the analysis of MN we use the following strategy. We
first investigate the \emph{benign} case when there is no super-efficient link. 
In the benign case, MN behaves like Nash and 
the worst case bound of 4/3 on the PoA can never be attained. More precisely,
we will exhibit a 
function $B(R_1,\ldots,R_{k-1})$ which is smaller than $4/3$ for all choices
of the $R_i$'s and will prove $\CMN(r) \le B(R_1,\ldots,R_{k - 1})
\Copt(r)$. We then investigate the non-benign case. We will derive
a recurrence relation for 
\[  \ePoA(R_1,\ldots,R_{k-1}) = \max_{r} \frac{\CM^{R_1,\ldots,R_{k - 1}}(r)}{\Copt(r)}. \]
In the case of a single link, i.e., $k = 1$, MN behaves like Nash which in
turn is equal to Opt. Thus \ignore{\Kurt{$\ePoA = 1$}}{$\ePoA() = 1$}. The coming subsections are devoted
to the analysis of two links and more than two links, respectively. 

\subsection{Two Links}\label{sec: two links}

The modified algorithm is determined by a parameter $R \ge 2$. If $\lambda_2 \le
R \lambda_1$, modified Nash is identical to Nash. If $\lambda_2 > R \lambda_1$,
the modified algorithm freezes the flow across the first link at 
$r_2/2$ once it reaches this level, i.e., $\hell_1(x) = \ell_1(x)$ for $x \le
r_2/2$ and $\hell_1(x) = \infty$ for $x > r_2/2$.\footnote{In Pigou's network we have $\ell_1(x) = x$ and $\ell_2(x) = 1$. Thus
$\lambda_2 = \infty$.  The modified cost functions are $\hell_2(x) = \ell_2(x)$
and $\hell_1(x) = x$ for $x \le r_2/2 = 1/2$ and $\hell_1(x) = \infty$ for $x >
1/2$. The Nash flow with respect to the modified cost function is identical to the
optimum flow in the original network and $\CM(f^*) = C(f^*)$. Thus $\ePoA = 1$
for Pigou's network.}

\begin{theorem}\label{thm:two-links-simple} 
For the case of two links, $\ePoA \le \max\left\{1 + 1/R,(4 + 4R)/(4 + 3R)\right\}$. In particular
$\ePoA = 5/4$ for $R = 4$. 
\end{theorem}

\begin{proof} Consider first the benign case $\lambda_2 \le R \Lambda_1$. There are three regimes: 
for $r \le r_2/2$, Opt and Nash behave
identically. For $r_2/2 \le r \le r_2$, Opt uses both links and Nash uses only
the first link, and for $r \ge r_2$, Opt and Nash use both links. 
$\PoA(r)$ is increasing for $r \le r_2$ and
decreasing for $r \ge r_2$. The worst case is at $r = r_2$. 
\ignore{&= \frac{\CN(r)}{\Copt(r)} = \frac{\Copt(r_2/2) +
\frac{1}{\Lambda_1}\left((r - r_2/2)^2 + (\Gamma_1 + r_2)(r - r_2/2)
\right)}{\Copt(r_2/2) + \frac{(r - r_2/2)^2}{\Lambda_2} + b_2(r - r_2/2)} \\
&= \frac{\Copt(r_2/2) +
\frac{(r - r_2/2)^2}{\Lambda_1} + b_2 (r - r_2/2)}{\Copt(r_2/2) + \frac{(r -
r_2/2)^2}{\Lambda_2} + b_2(r - r_2/2)}.
\end{align*}
The last equality uses $\Gamma_1 + r_2 = b_2
\lambda_1$. By Lemma~\ref{lem: derivative}, the sign of the derivative is equal
to the sign of $\frac{1}{\Lambda_1} - \frac{1}{\Lambda_2}$ and hence is
non-negative. 

Assume $r \ge r_2$. Then
\[ PoA(r) = \frac{\CN(r)}{\Copt(r)} = \frac{\CN(r)}{\CN(r) - C_2} = \frac{1}{1
- C_2/\CN(r)}\]
and hence $\PoA(r)$ is a decreasing function of $r$. The worst case occurs for
$r = r_2$. We have thus shown that $\PoA(r)$ assumes its maximum at $r_2$.} 
Then $\PoA(r_2) = \CN(r_2)/\Copt(r_2) = \CN(r_2)/(\CN(r_2) - C_2) = 1/(1 - C_2/\CN(r_2))$. 
We upper-bound $C_2/\CN(r_2)$. Recall that $r_2 = (b_2 -
b_1)\lambda_1$, $r_2 + \Gamma_1 = b_2 \lambda_1$ and $\CN(r_2) =
1/\lambda_1(r_2^2 + \Gamma_1 r_2)$. We obtain 
\[ \frac{C_2}{\CN(r_2)} = \frac{(b_2 - b_1)^2 \lambda_1 \lambda_2}{4 \Lambda_2
(1/\lambda_1) (r_2^2 + \gamma_1 r_2)} 
= \frac{(b_2 - b_1)^2 \lambda_1 \lambda_2}{4 \Lambda_2
(1/\lambda_1) (b_2 - b_1) \lambda_1 b_2 \lambda_1}
\le \frac{\lambda_2}{4 \Lambda_2} \le \frac{1}{4 (1 + 1/R)}.
\]
Thus $\PoA(r) \le B(R) \assign 
\frac {1} {1 - \frac R {4(R+1)}} 
%{1}/\left(1 - R/(4(R+1))\right) 
= (4 + 4R)/(4 + 3R)$. 

We come to the case $\lambda_2 > R \Lambda_1$:
There are two regimes: for $r \le r_2/2$, Opt and MN behave
identically. For $r > r_2/2$, Opt uses both links and MN routes
$r_2/2$ over the first link and $r - r_2/2$ over the second link. Thus for $r
\ge r_2/2$:
\[ \ePoA(r) = \frac{\CMN(r)}{\Copt(r)} = \frac{\Copt(r_2/2) +
\frac{(r - r_2/2)^2}{\lambda_2} + b_2(r - r_2/2)}
{\Copt(r_2/2) + \frac{(r - r_2/2)^2}{\Lambda_2} + b_2(r - r_2/2)} \le
\frac{\Lambda_2}{\lambda_2} \le 1 + 1/R.\] 
 \end{proof}

\subsection{Many Links}\label{sec: many links} 

As already mentioned, we distinguish cases. We first study the benign case 
$\lambda_{i+1} \le R_i \Lambda_i$ for all $i$, $1 \le i < k$, and then deal
with the non-benign case. 

\paragraph{The Benign Case:}  We assume $\lambda_{i+1} \le R_i \Lambda_i$ for all
$i$, $1 \le i < k$. Then MN behaves like Nash. We will show $e\PoA \le
B(R_1,\ldots,R_{k-1}) < 4/3$; here $B$ stands for benign case or base case. 
Our proof strategy
is as follows; we will first show (Lemma~\ref{lem:upper_bound_ratio_flows})
that for the $i$-th link the ratio of Nash flow to optimal flow
is bounded by $2 \Lambda_k/(\Lambda_i + \Lambda_k)$. This ratio is never more
than two; in the benign case, it is bounded away from two. We will then use
this fact to derive a bound on the Price of Anarchy (Lemma~\ref{lem:upper-bound-PoA1}).

\begin{lemma}\label{lem:upper_bound_ratio_flows}
Let $h$ be the number of links that Opt is using. Then 
\[ \frac{f_i^N}{f_i^*}\leq\frac{2\Lambda_h}{\Lambda_i+\Lambda_h}\]
for $i \le h$. 
	If $\lambda_{i'+1} \le R_{i'} \Lambda_{i'}$ for all $i'$, then 
\[ \frac{2\Lambda_h}{\Lambda_i+\Lambda_h} \le \frac{2 P}{P + 1},\]
where $P := \prod_{1 \leq i < k} (1 + R_i)$. 
\end{lemma}

\begin{proof} 
Let $j$ be the number of links that Nash is using.
For $i > j$, the Nash flow on the $i$-th link is zero and the claim is obvious.
For $i \le j$, we can write the Nash and the optimal flow through link $i$ as
\[ f_i^N = r \lambda_i/\Lambda_j + (\Gamma_j \lambda_i
/\Lambda_j - \gamma_i) \quad\text{and}\quad 
f_i^* = r \lambda_i/\Lambda_h + (\Gamma_h \lambda_i
/\Lambda_h - \gamma_i)/2.\]
Therefore their ratio as a function of $r$ is
\begin{align*}
F(r)=\frac{f_i^N}{f_i^*} = \frac{\Lambda_h}{\Lambda_j}\cdot \frac{2r+2\Gamma_j-2b_i\Lambda_j}{2r+\Gamma_h-b_i\Lambda_h}.
\end{align*}
The sign of the derivative $F'(r)$ is equal to the sign of
$\Gamma_h-b_i\Lambda_h-2\Gamma_j+2b_i\Lambda_j$ and hence constant. Thus $F(r)$
attains its maximum either for $r_j$ or for $r_{j+1}$. We have 
\begin{align*}
F(r_{j+1}) &\leq \frac{\Lambda_h}{\Lambda_j}\cdot
\frac{2r_{j+1}+2\Gamma_j-2b_i\Lambda_j}{2r_{j+1}+\Gamma_h-b_i\Lambda_h}
= \frac{\Lambda_h}{\Lambda_j}\cdot
\frac{2(b_{j+1}-b_i)\Lambda_j}{2b_{j+1}\Lambda_{j+1}-2\Gamma_{j+1}+\Gamma_h-b_i\Lambda_h}\\
& = \frac{2(b_{j+1}-b_i)\Lambda_h}{\sum_{g\leq
    j+1}(2b_{j+1}-2b_g)\lambda_g+\sum_{g\leq h}(b_g-b_i)\lambda_g}
= \frac{2(b_{j+1}-b_i)\Lambda_h}{\sum_{g\leq
    j}(2b_{j+1}-b_g-b_i)\lambda_g+\sum_{j<g\leq
    h}(b_g-b_i)\lambda_g}\\
& = \frac{2(b_{j+1}-b_i)\Lambda_h}{\sum_{g\leq
    i}(2b_{j+1}-b_g-b_i)\lambda_g+\sum_{i<g\leq
    j}(2b_{j+1}-b_g-b_i)\lambda_g+\sum_{j<g\leq
    h}(b_g-b_i)\lambda_g}\\
& \leq \frac{2(b_{j+1}-b_i)\Lambda_h}{\sum_{g\leq
    i}2(b_{j+1}-b_i)\lambda_g+\sum_{i<g\leq
    h}(b_{j+1}-b_i)\lambda_g}\\
& = \frac{2\Lambda_h}{\sum_{g\leq
    i}2\lambda_g+\sum_{i<g\leq h}\lambda_g} = \frac{2\Lambda_h}{\Lambda_i+\Lambda_h}
\end{align*}
and 
\begin{align*}
F(r_j) &\leq \frac{\Lambda_h}{\Lambda_j}\cdot
\frac{2r_{j}+2\Gamma_j-2b_i\Lambda_j}{2r_{j}+\Gamma_h-b_i\Lambda_h}
= \frac{\Lambda_h}{\Lambda_j}\cdot
\frac{2(b_{j}-b_i)\Lambda_j}{2b_{j}\Lambda_{j}-2\Gamma_{j}+\Gamma_h-b_i\Lambda_h}\\
& = \frac{2(b_{j}-b_i)\Lambda_h}{\sum_{g\leq
    j}(2b_{j}-2b_g)\lambda_g+\sum_{g\leq h}(b_g-b_i)\lambda_g}
= \frac{2(b_{j}-b_i)\Lambda_h}{\sum_{g\leq
    j}(2b_{j}-b_g-b_i)\lambda_g+\sum_{j<g\leq
    h}(b_g-b_i)\lambda_g}\\
& = \frac{2(b_{j}-b_i)\Lambda_h}{\sum_{g\leq
    i}(2b_{j}-b_g-b_i)\lambda_g+\sum_{i<g\leq
    j}(2b_{j}-b_g-b_i)\lambda_g+\sum_{j<g\leq
    h}(b_g-b_i)\lambda_g}\\
& \leq \frac{2(b_{j}-b_i)\Lambda_h}{\sum_{g\leq
    i}2(b_{j}-b_i)\lambda_g+\sum_{i<g\leq
    h}(b_{j}-b_i)\lambda_g}
= \frac{2\Lambda_h}{\sum_{g\leq
    i}2\lambda_g+\sum_{i<g\leq h}\lambda_g} = \frac{2\Lambda_h}{\Lambda_i+\Lambda_h}.
\end{align*}

If $\lambda_{i'+1} \le R_{i'} \Lambda_{i'}$ for all $i'$, then $\Lambda_{i'+1} =
\lambda_{i'+1} + \Lambda_{i'} \le (1 + R_{i'}) \Lambda_{i'}$ for all $i'$ and hence
$\Lambda_{h} \le \Lambda_k  \le  P \Lambda_i$. 
 \end{proof}

\begin{lemma}\label{lem:arithmetic-inequality}
For any positive reals $\mu$, $\alpha$, and $\beta$ with $1 \le \mu \le 2$ and
$\alpha/\beta \le \mu$, $\beta \alpha \le \frac{\mu - 1}{\mu^2} \alpha^2 +
\beta^2$. 
\end{lemma}

\begin{proof} 
We may assume $\beta \ge 0$. If $\beta = 0$, there is nothing to show. 
So assume $\beta > 0$ and let $\alpha/\beta = \delta \mu$ for some $\delta
\le 1$. We need to show (divide the target inequality by $\beta^2$) 
$\delta \mu \le (\mu - 1) \delta^2 + 1$ or equivalently $\mu \delta (1 - \delta)
\le (1 - \delta) (1 + \delta)$. This inequality holds for $\delta \le 1$ and
$\mu \le 2$.
 \end{proof}

\begin{lemma}\label{lem:upper-bound-PoA1}
If $f^N_i/f_i^*\leq \mu \le 2$ for all $i$, then $\PoA\leq {\mu^2}/(\mu^2-\mu+1)$. 
If $\lambda_{j+1} \le R_j \Lambda_j$ for all $j$, then 
\[   \PoA \le B(R_1,\ldots,R_{k-1}) \assign \frac{4 P^2}{3 P^2 + 1} ,\]
where $P := \prod_{1 \leq i < k} (1 + R_i)$. 
\end{lemma}

\begin{proof} 
Assume that Nash uses $j$ links and
let $L$ be the common latency of the links used by Nash. Then $L = a_i f_i^N +
b_i$ for $i \le j$ and $L \le b_{i} = a_i f_i^N + b_i $ for $i > j$. Thus, 
by use of Lemma~\ref{lem:arithmetic-inequality},
\begin{align*}
\CN(r) &= Lr = \sum_i L f_i^* \le \sum_i \left(a_i f_i^N + b_i\right) f_i^* 
      \le \frac{\mu-1}{\mu^2}\sum_i a_i (f_i^N)^2 + \sum_i \left(a_i (f_i^*)^2 +
b_i f_i^*\right)\\
&\le \frac{\mu-1}{\mu^2} \CN(r) + \Copt(r) \end{align*}
and hence $\PoA\leq {\mu^2}/(\mu^2-\mu+1)$. If $\lambda_{j+1} \le R_j
\Lambda_j$ for all $j$, employing
Lemma~\ref{lem:upper_bound_ratio_flows}, we may use $\mu = 2P/(P + 1)$
and obtain $\PoA \le 4P^2/(3P^2 + 1)$. 
 \end{proof}

\paragraph{The General Case:} We come to the case where $\lambda_{i+1} \ge R_{i}
\Lambda_{i}$ for some $i$. Let $j$ be the smallest such $i$. 
For $r \le r_{j+1}/2$, MN and Opt use only links $1$ to $j$ and we are in the
benign case. Hence $e\PoA$ is bounded by $B(R_1,\ldots,R_{j-1}) < 4/3$. 
Assume now that $r>r_{j+1/2}$. MN
routes the flow exceeding $r_{j+1}/2$ exclusively on higher links. 

\begin{lemma} \label{lem:MNbeforeOPT} MN does not use links before Opt. \end{lemma}

\begin{proof} 
This is trivially true for the $j+1$-st link. Consider any $h > j+1$. MN starts to use link $h$ at 
$s_h = r_{j+1}/2 + \sum_{j+1 \le i < h}(b_{i+1} - b_i) (\Lambda_i -
\Lambda_j)$
and Opt starts to use it at $r_h/2 = r_{j+1}/2 + \sum_{j+1 \le i < h}(b_{i+1} -
b_i)\Lambda_i/2$. We have $s_h \ge r_h/2$ since $\Lambda_i - \Lambda_j \ge
\Lambda_i/2$ for $i > j$.  
 \end{proof}

We need to bound the cost of MN in terms of the cost of Opt. In order to do so,
we introduce an intermediate flow Mopt (modified optimum) that we can readily relate to MN and to
Opt. Mopt uses links $1$ to $j$ to route
$r_{j+1}/2$ and routes $f = r- r_{j+1}/2$ optimally across links $j+1$ to
$k$. Let $f_i^*$ and $f_i^m$ be the
optimal flows and the flows of Mopt, respectively, at rate $r$.
Let $r_s = \sum_{i \le j} f_i^* \ge r_{j+1}/2$ be the total flow routed across the first $j$
links in the optimal flow (the subscript $s$ stands for small) and let
\[    t  = \frac{r - r_{j+1}/2}{r - r_s}\]
We will show $t \le 1 + 1/R_j$ below. We next relate the cost of Mopt on links
$j+1$ to $k$ to the cost of Opt on these links. To this end we scale the
optimal flow on these links by a factor of $t$, i.e., we consider 
the following flow across links $j+1$ to $k$: on link $i$, $j+1 \le i
\le k$, it routes $t \cdot f_i^*$. The total flow on the \emph{high} links,
i.e., links $j+1$ to $k$, is $r - r_{j+1}/2$ and hence Mopt incurs at most the
cost of this flow on its high links. Thus
\[   \sum_{i > j} \ell_i(f_i^m) f_i^m \le \sum_{i > j} \ell_i(t f_i^*)
t f_i^* \le  t^2 \left(\sum_{i > j} \ell_i(f_i^*)
f_i^*\right). \]
The cost of MN on the high links is at most $\ePoA(R_{j+1},\ldots,R_{k-1})$ times
this cost by the induction hypothesis. We can now bound the cost of MN as
follows: 
% (Kurt changes this on September 1st; I left the old version in the
% document; the old version comes first, the new version comes then).
% 
% \begin{align*}
% \CMN(r) &= \CN(r_{j+1}/2) + \CMN(\text{flow $f$ across links $j+1$ to $k$})\\
%         &\le B(R_1,\ldots,R_{j-1}) \Copt(r_{j+1}/2) + t^2 \ePoA(R_{j+1},\ldots,R_{k-1})\left(\sum_{i > j} \ell_i(f_i^*)
% f_i^*\right) \\
%         &\le B(R_1,\ldots,R_{j-1}) \Copt(r_s) + t^2 \ePoA(R_{j+1},\ldots,R_{k-1})\left(\sum_{i > j} \ell_i(f_i^*)
% f_i^*\right) \\
% &\le \max\{B(R_1,\ldots,R_{j-1}), t^2 \ePoA(R_{j+1},\ldots,R_{k-1})\} \Copt(r)
% \end{align*}
% 
\begin{align*}
\CMN(r) &= \CN(r_{j+1}/2) + \CMN(\text{flow $f$ across links $j+1$ to $k$})\\
        &\le B(R_1,\ldots,R_{j-1}) \Copt(r_{j+1}/2) + t^2 \ePoA(R_{j+1},\ldots,R_{k-1})\left(\sum_{i > j} \ell_i(f_i^*)
f_i^*\right) \\
        &\le B(R_1,\ldots,R_{j-1})\left(\sum_{i \le j} \ell_i(f_i^*)
f_i^*\right) + t^2 \ePoA(R_{j+1},\ldots,R_{k-1})\left(\sum_{i > j} \ell_i(f_i^*)
f_i^*\right) \\
&\le \max\left\{B(R_1,\ldots,R_{j-1}),\, t^2 \ePoA(R_{j+1},\ldots,R_{k-1})\right\} \Copt(r)
\end{align*}

% The proof of Lemma~\ref{lem:tbound} is omitted due to space limitations. 
% Lemma~\ref{thm: recurrence for PoA} summarizes the discussion and 
% Theorem~\ref{thm:ePoAbound} gives the main result of this section.

\begin{lemma} \label{lem:tbound} $t \le 1 + 1/R_{j}$,
where $j$ is the smallest $i$ for which $\lambda_{i+1} \ge R_{i} \Lambda_{i}$. 
\end{lemma}

\begin{proof} 
Assume that Opt uses $h$ links where $j+1 \le h \le k$. Then $r_h/2 \le r \le
r_{h+1}/2$. Let $r = r_h/2 + \delta$. According to Theorem~\ref{thm: basic facts}, $f_i^* = r
\lambda_i/\Lambda_h + (\Gamma_h \lambda_i/\Lambda_h - \gamma_i)/2$ and hence 
\[  r_s = \left(\frac{r_h}{2} + \delta \right) \frac{\Lambda_j}{\Lambda_h} + \frac{1}{2} \left(\frac{\Gamma_h
\Lambda_j}{\Lambda_h} - \Gamma_j\right).\]
Since $\Gamma_h + r_h = b_h\Lambda_h$ and  $\Gamma_j + r_j = b_j\Lambda_j$ (see
Theorem~\ref{thm: basic facts}), this
simplifies to 
\[ r_s = \frac{\Lambda_j \delta}{\Lambda_h} + \frac{b_h \Lambda_h - \Gamma_h}{2} \frac{\Lambda_j}{\Lambda_h} + \frac{1}{2} \left(\frac{\Gamma_h
\Lambda_j}{\Lambda_h} - b_j \Lambda_j + r_j \right) = \frac{\Lambda_j \delta}{\Lambda_h} + \frac{1}{2} \left((b_h - b_j)
\Lambda_j + r_j\right) = \frac{\Lambda_j \delta}{\Lambda_h} + r^*_s,\]
where $r^*_s = \frac{1}{2} \left((b_h - b_j)
\Lambda_j + r_j\right)$. We can now bound $t$. 
\[ t = \frac{r- r_{j+1}/2}{r - r_s} = \frac{r_h/2 + \delta - r_{j+1}/2}{r_h/2 +
\delta - r^*_s - \Lambda_j \delta /\Lambda_h} \le \max \left\{ \frac{r_h/2 -
r_{j+1}/2}{r_h/2 - r^*_s}, \frac{1}{1 -
\frac{\Lambda_j}{\Lambda_h}}\right\}.\]
Next observe that 
\begin{align*}
 \frac{r_h/2 -
r_{j+1}/2}{r_h/2 - r^*_s}& = \frac{\sum_{j+1 \le i < h} (b_{i+1} - b_i) \Lambda_i}{\left(\sum_{j \le i <
h} (b_{i+1} - b_i) \Lambda_i\right) - (b_h - b_j) \Lambda_j}
= \frac{\sum_{j+1 \le i < h} (b_{i+1} - b_i) \Lambda_i}{\sum_{j + 1\le i <
h} (b_{i+1} - b_i) (\Lambda_i - \Lambda_j)}\\
&\le \max_{j+1 \le i < h} \frac{\Lambda_i}{\Lambda_i - \Lambda_j}
= \frac{\Lambda_{j+1}}{\Lambda_{j+1}- \Lambda_j}
= \frac{\Lambda_{j} + \lambda_{j+1}}{\lambda_{j+1}}
\le 1 + \frac{1}{R_{j}}.
\end{align*}
The second term in the upper bound for $t$ is also bounded by this quantity. 
 \end{proof}

We summarize the discussion.

\begin{lemma}\label{thm: recurrence for nonbenign PoA}
%\label{thm: recurrence for PoA} 
For every $k$ and every $j$ with
$1 \le j < k$. If $\lambda_{j+1} > R_{j} \Lambda_j$ and 
$\lambda_{i+1} \le R_i \Lambda_i$ for $i < j$, then 
\[  \ePoA(R_1,\ldots,R_{k-1}) \le \max\left\{B(R_1,\ldots,R_{j-1}),\left(1 + 
\frac{1}{R_{j}}\right)^2\ePoA(R_{j+1},\ldots,R_{k-1})\right\}.\]
\end{lemma}

We are now ready for our main theorem.

\begin{theorem}\label{thm:ePoAbound}
For any $k$, there is a choice of the parameters $R_1$ to
$R_{k-1}$ such that the engineered Price of Anarchy  with these
parameters is strictly less than $4/3$. \end{theorem}
\begin{proof} We will show $\ePoA(R_i,\ldots,R_{k-1}) < 4/3$ by downward
  induction on $i$, i.e., we will define $R_{k-1}$, $R_{k-2}$, down to $R_1$ in this
  order. For $i = k$, we have $\ePoA() = 1 < 4/3$. 

  We now come to the induction step. We have already defined $R_{k-1}$
  down to
  $R_{i+1}$ and now define $R_i$. We have 
\[ 
\ePoA(R_i,\ldots,R_{k-1}) \le \max \left\{ \begin{array}{l} B(R_i,\ldots,R_{k-1}), \\
        \max_{j; i \le j < k} \left\{ B(R_i,\ldots,R_{j-1}),\left(1 + 
\frac{1}{R_{j}}\right)^2\ePoA(R_{j+1},\ldots,R_{k-1}) \right\}
\end{array}\right\}, \]
where the first line covers the benign case and the second line covers the
non-benign case (Lemma~\ref{thm: recurrence for nonbenign PoA}). We now fix
$R_i$. We have $B(R_i,\ldots,R_{k-1})< 4/3$ and
$B(R_i,\ldots,R_{j-1}) < 4/3$ for all $j$, $i \le j < k$ by Lemma~\ref{lem:upper-bound-PoA1} for any
choice of $R_i$. This completes the induction if the first line defines the
maximum. So assume that the second line defines the maximum. We only need to
deal with the case $j = i$ in the second line, as the case $j > i$ was already
dealt with for larger $i$. The case $j = i$ is handled by choosing $R_i$
sufficiently large, i.e., such that $(1 + 1/R_i)^2
\ePoA(R_{i+1},\ldots,R_{k-1}) < 4/3$. 
 \end{proof}

\begin{remark}Alternatively, we could take $i = k-1$ as the base case. Then, $j = i-1$
in the non-benign case and hence 
\begin{align*}
\ePoA(R_{k-1}) & = \max (B(R_{k-1}), B(), (1 + \frac{1}{R_{k-1}})^2 \ePoA() )\\
                         & = \max(\frac{4 (1 + R_{k-1})^2}{3 (1 + R_{k-1})^2 + 1}, (1 + \frac{1}{R_{k-1}})^2)
                         \end{align*}
where the bound on $B(R_{k-1})$ comes from Lemma 4. For $R_{k-1} = 7$, the
bound becomes $\max((4/3)(64/65),64/49) = \max(64/65,48/49)\cdot 4/3$. 
\end{remark}

\section{An Improved Mechanism for the Case of Two Links}
\label{sec: two links advanced}

In this section we present a mechanism which achieves $\npoa=1.192$ for a
network that consists of two parallel links. 
The ratio $\CN(r)/\Copt(r)$ is maximized at $r
= r_2$. At this rate Nash still uses only the first link and Opt uses both
links. In order to avoid this maximum ratio (if larger than 1.192), 
we force MN to use the second link earlier
by increasing the latency of the first link after some rate $x_1$, $r_2/2 \le x_1
\le r_2$ to a value above $b_2$. In the preceding section, we increased the
latency to $\infty$. In this way, we avoided a bad ratio at $r_2$, but paid a
price for very large rates. The idea for the improved construction, is to
increase the latency to a finite value. This will avoid the bad ratio, but also
allow MN to use both links for large rates. 
In particular, we obtain the following result.
\begin{theorem}\label{thm:two-links-advanced}
There is a mechanism for a network of two parallel links that achieves
$\npoa=1.192$.
\end{theorem}

\begin{proof} 
Recall (Theorem~\ref{thm:two-links-simple}) that the Price of Anarchy is upper bounded by $(4 + 4R)/(4 +
3R)$ where $R = a_1/a_2$. Let $R_0$ be such that $(4 + 4R_0)/(4 + 3R_0) =
1.192$. Then $R_0 = 96/53$. We only need to consider the case $R >
R_0$. 
The latency function of the second link is unchanged and the latency
function of the first link is changed into 
\begin{equation}
\widehat{\ell}_1(x)= \left\{ 
		\begin{array} {lll}
		\ell_1(x), &		          x\leq x_1	\\ 
		\ell_1(x_2), &	               x_1 <x \le x_2	\\ 
		\ell_1(x) &               x> x_2.  
		\end{array}
	\right.
\label{eq:mechanism_2links}
\end{equation}
where $x_1$ and $x_2$ satisfy $r_2/2 \le x_1 \le r_2 \le x_2$ and will be fixed
later. In words, when either the flow in the first link does not exceed
$x_1$, or is larger than $x_2$, the network remains unchanged.  
However, when the flow in the first link is between these two
values, the mechanism increases the latency of this link to  $\ell_1(x_2)$.  
Let $r^*$ be such that 
\[         \ell_2(r^* - x_1) = \ell_1(x_2). \]
We will fix $x_1$ and $x_2$ such that $r^* \ge r_2$.

What is the effect of this modification? For $r \le r_2/2$, Opt and MN are
the same and $\ePoA(r) = 1$. For $r_2/2 \le r \le x_1$, MN behaves like Nash
and $\ePoA(r)$ increases. 
At $r = x_1$, MN starts to use the second link. 
MN will route any additional flow on the second link until $r = r^*$. 
At $r = r^*$, MN routes $x_1$ on the
first link and $r^* - x_1$ on the second link. Beyond $r^*$, MN routes additional flow on the
first link until the flow on the first link has grown to $x_2$. This is the
case at $r^{**} = r^* - x_1 + x_2$. For $r \ge r^{**}$, MN behaves like
Nash. 

\begin{figure}[t] 
\begin{center}
%\psfrag{ePoA}{$\ePoA$}\psfrag{r}{$r$}\psfrag{r22}{$r_2/2$}\psfrag{r2}{$r_2$}\psfrag{r}{$r$}\psfrag{rs}{$r^*$}\psfrag{rss}{$r^{**}$}\psfrag{x1}{$x_1$}
%\epsfig{file=ePoA-2-links,width=0.7\textwidth}
 \begin{tikzpicture}[auto, right,  node distance=0.5cm, xscale=2.0]
	\tikzstyle{axis}=[->,draw=black,line width=1pt,rounded corners]
	\tikzstyle{graph}=[draw=black,line width=1pt]
	\tikzstyle{graph-blue}=[dotted,draw=black,line width=1pt]
	
%	\node (ph1) at (-1.5cm, -0.5cm) {};
%	\node (ph2) at (7.0cm, 5.0cm) {};

	\coordinate(origin) at (0,0) {};
	\coordinate(xmax) at (6.5cm,0) {};
	\coordinate(ymax) at (0,5cm) {};

	\draw[axis] (origin.center) -- (xmax); 
	\draw[axis] (origin.center) -- (ymax);

	\node[anchor=east] (ylabel) at ([xshift=-1ex, yshift=-2ex]ymax.east) {$\npoa(r)$};
	\node[anchor=north] (xlabel) at ([xshift=0ex, yshift=-1ex]xmax.south) {$r$};

	\coordinate(ropt) at (1.3,1) {};

	\draw[graph] (0,1) coordinate (graphbegin)  --  (ropt);  
	\draw[graph-blue] (1.3,0) coordinate (xax1) --(ropt);	

	\node[anchor=east]  at ([xshift=-1ex]graphbegin.east) {1};

	\draw[graph] (ropt)  
			.. controls (1.02,1.0) and (1.7,1.0)..
			(2.1cm,1.7) coordinate (P1) 
			.. controls (2.2,1.9) and (2.31,2.39)..
			(2.32cm,2.4) coordinate (P2) ;

	\draw[graph] (P1) .. controls (2.2,1.9) and (2.31,2.39).. (P2)
			.. controls (2.5,1.7) and (2.7,0.4)..
			(3.1,1.5) coordinate (Q2);
	\draw[graph] ([yshift=1.1cm]Q2) coordinate (Q2top)
			.. controls (3,2.6cm) and (3.9cm,1.0cm)..
			(6,1.0cm) coordinate (limit);

	\node[anchor=center]  at ([yshift=-1.5ex]xax1.south) {${r_2/2}$};
	\draw[graph-blue] (2.32cm,0) coordinate (x1) --(P2);	
	\node[anchor=center] (x1)  at ([yshift=-1.5ex]x1.south) {$x_1$};
	\draw[graph-blue] (3.1cm,0) coordinate (x2) --(Q2top);	
	\node[anchor=center] (rstar)  at ([yshift=-1.5ex]x2.south) {$r^*$};
 	\node (rnash) at (2.6cm,0) {};
	\node[anchor=center] at ([yshift=-1.5ex]rnash) {$r_2$};

	\draw[graph-blue] (3.7cm,0) coordinate (x3) --(3.7cm,1.9cm);	
	\node[anchor=center] (rstar)  at ([yshift=-1.5ex]x3.south) {$r^{**}$};

	\end{tikzpicture}
\end{center}
	\caption{The engineered price of anarchy for the construction of
	Section~\ref{sec: two links advanced}.\label{fig: ePoA for 2 links}}
\end{figure}
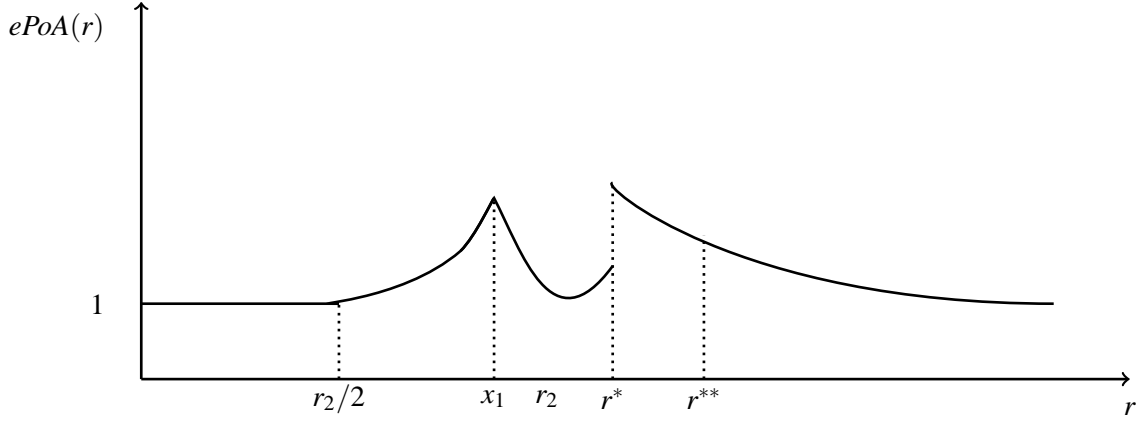

Figure~\ref{fig: ePoA for 2 links} shows the graph of $\ePoA(r)$. We have
$\ePoA(r) = 1$ for $r \le r_2/2$. For $r_2/2 \le r \le x_1$, $\ePoA(r)$ increases
to 
\[              \ePoA(x_1) = \frac{a_1 x_1^2 + b_1 x_1}{\Copt(x_1)}.\]
For $x_1 \le r \le r^*$, $\ePoA(r)$ is convex. It will first decrease and reach
the value one (this assumes that $r^*$ is big enough) at the rate where Opt
routes $x_1$ on the first link; after this rate it will increase again. 
At $r^*$, $\ePoA$ has a
discontinuity because at $r^*$ MN routes $x_1$ on the first link for a cost of
$\ell_1(x_1) x_1$ and at $r^*+\epsilon$ it routes $x_1+\epsilon$ on the first link for a cost
of $\ell_1(x_2) (x_1+\epsilon)$. Thus 
\[        \lim_{r \rightarrow r^*_+} \frac{\CMN(r)}{\Copt(r)} = 
\lim_{r \rightarrow r^*_+}\frac{\ell_1(x_2)
r}{\Copt(r)} = \frac{\ell_1(x_2) r^*}{\Copt(r^*)} = \frac{\ell_2(r^* - x_1) r^*}{\Copt(r^*)}.\]
For $r \ge r^*$, $\ePoA(r)$ decreases. Thus 
\begin{equation}
\ePoA = \max\left\{\frac{a_1 x_1^2 + b_1 x_1}{\Copt(x_1)}, \frac{\ell_2(r^* - x_1)
r^*}{\Copt(r^*)}\right\}. \label{eq: upper bound} \end{equation}
It remains to show that $x_1\le r_2$ and $r^*\ge r_2$ can be
chosen\footnote{The optimal choice for $x_1$ and $r^*$ is such that both terms
are equal and as small as possible. 
We were unable to solve the resulting system explicitly. We will prove in the
next section that the
mechanism defined by these optimal choices of the parameters $x_1$ and $r^*$ is
optimal.} such that the right-hand
side is at most 1.192. By Theorem~\ref{thm: basic facts}~(c), 
\[     \Copt(r) = b_1r + \frac{a_1}{1 + R}\left(r^2 + R r_2 r - R r_2^2/4\right), \]
for $r \ge r_2/2$ and $R = a_1/a_2$. Also $\ell_2(r^* - x_1) = a_2(r^* - x_1) +
b_2 = a_2(r^* - x_1) + b_1 + a_1 r_2$. 

We first determine the maximum $x_1 \le r_2$ such that $\left(a_1 x_1^2 + b_1
x_1\right)/\Copt(x_1) \le 1.192$ for all $b_1$. Since the expression $\left(a_1 x_1^2 + b_1 x_1\right)/\Copt(x_1)$
is decreasing in $b_1$, this $x_1$ is determined for $b_1 = 0$. It follows that
$\alpha = x_1/r_2$ is defined by the equation
\begin{equation} \frac {4 (R +1) \alpha ^2} {4 \alpha  R - R +4 \alpha^2} = 1.192.
\label{eq: ePoA1}
	\end{equation}
For $R \ge R_0$, this equation has a unique solution $\alpha_0 \in [1/2,1]$,
namely $$\alpha_0= \frac 1 2\cdot \frac {149 R+2 \sqrt{894 R(R+1)}} {125
R-24}.$$ We turn to the second term in equation (\ref{eq: upper bound}). For
$r^* > r_2$, it is a decreasing function of $b_1$. Substituting $b_1 = 0$ into the second
term and setting $\beta = r^*/r_2$ yields after some computation
\begin{equation}
\label{eq:p2_d0}
    \ePoA_2 = \frac {4\beta (R+1) (\beta-\alpha+R) } {R(4 \beta^2+4 \beta R-R)}.
\end{equation}

\noindent For fixed $\alpha = \alpha_0$ and any $R \ge R_0$, $\ePoA_2$ is
minimized for $\beta = \beta_0 = \frac {R+\sqrt R\sqrt{R +4 \alpha_0
    (R -\alpha_0) }}{ 4\alpha_0}.$
%, where 
%\begin{equation}
%\label{eq:h_val}
%$ \beta_0 = \frac {R+\sqrt R\sqrt{R +4 \alpha_0 (R -\alpha_0)) }}{ 4\alpha_0}.$
%\end{equation}
For $R \ge R_0$, one can prove $\beta_0\ge 1$, as
needed. Substituting $\alpha_0$ and $\beta_0$ into $\ePoA_2$ yields a function
of $R$. It is easy to see, using the derivative, that the maximum value of this function
for  $R \ge R_0$ is at most $1.192$. 
 \end{proof}

%\input{lowerbound}

\iffalse
\begin{figure}[t]
\begin{center}
%\epsfig{file=PoA-ePoA.eps,angle=270,width=0.7\textwidth}
\end{center}
\caption{The plot shows the price of anarchy of the unmodified Nash and
engineered price of anarchy (with the scheme of Section TODO) as a function of
$R$. For $R$ up to approximately 1.6, modified Nash is equal to Nash. For
larger $R$, the modification decreases the price of anarchy. However, $\ePoA$
is still raising as a function of $R$. }
\end{figure}

\fi

\section{A Lower Bound for the Case of Two Links}\label{lower bound}

We prove that the construction of the previous section is optimal
among the class of deterministic mechanisms that guarantee the existence of an equilibrium for every rate $r>0$ and that use
non-decreasing\footnote{It remains open whether similar arguments can be
  applied for showing the lower bound for non-monotone mechanisms with respect to User Equilibria.} latency functions.
 %{\em even for the weaker notion of User Equilibria}.  
For these mechanisms we show that  $\ePoA \ge 1.191$. 

As in the preceding sections, we use $\hl$ to denote the modified latency
functions. The use of $R$ throughout this section denotes the ratio of
the linear coefficients of the two latency functions of the instance, and should not be confused with its use in the
previous sections, were it was a parameter of the mechanism. As mentioned above, we are making two assumptions about the $\hl$'s: an equilibrium
flow must exist for every rate $r$, and $\hl_i$  is non-decreasing, i.e., 
if $x<x'$, then $\hl_i(x) \le \hl_i(x')$, for $i=1,2$. 
It is worthwhile to recall the equilibrium 
conditions for general latency functions (as given by Dafermos-Sparrow~\cite{DS69}): 
if $(x,y)$ is an equilibrium for
rate $r = x + y$, then
$\hl_2(y') \ge \hl_1(x)$ for $y' \in (y,r]$ (otherwise $y' - y$ amount of flow
would move from the first link to the second) and $\hl_1(x') \ge \hl_2(y)$
for $x' \in (x,r]$ (otherwise, $x' - x$ amount of flow would move from the
second link to the first). Since we assume our functions to be monotone, the
condition $\hl_2(y') \ge \hl_1(x)$ for $y' \in (y,r]$ is equivalent to 
$\lim \inf_{y' \downarrow y} \hl_2(y') \ge \hl_1(x)$ provided that $y < r$ (or
equivalently, $x > 0$). Since we are discussing a network of two parallel links,
the latter condition is in turn equivalent to (\ref{user equilibrium}). 

\begin{theorem}\label{thm: lower bound for 2links} The construction of
Section~\ref{sec: two links advanced} is optimal and $\ePoA \ge 1.191$. 
\end{theorem}

\begin{proof}
We analyze a network with 
latency functions $\ell_1(x) = x$ and $\ell_2(x) = x/R +1 = (x + R)/R$, $2 \le R \le
4$, and
derive a lower bound as a function of the parameter $R$; the restriction $2 \le R \le 4$ will
become clear below. In a second step we
choose $R$ so as to maximize the lower bound; the optimal choice is $R = R^*
\approx 2.1$. For $r \le 1/2$, Opt uses one link and $\Copt(r) = r^2$, and for
$r \ge 1/2$, Opt uses two links and 
$\Copt(r) = (r^2 + Rr - R/4)/(1 + R)$; $\Copt(1) = (3R + 4)/(4R + 4)$ and
$\CN(1) = 1$. Thus $\PoA = \PoA(1) = (4 R + 4)/(3 R + 4)$. For $R \ge 2$, we
have $\PoA(1) \ge 12/10 = 1.2$.

Consider now some modified latency functions $\hat \ell_1, \hat \ell_2$, and 
let $(x_1,1-x_1)$ be an equilibrium flow for rate $1$ for the modified network.
Let
\[ r^* = \inf \{ r \; ;\; \text{there is an equilibrium flow $(x,r - x)$ for MN with
$x > x_1$} \} ;\]
$r^* = \infty$ if there is no equilibrium flow $(x,y)$ with $x > x_1$. The
equilibrium conditions for flow $(x_1,1-x_1)$ imply 
\begin{equation}\label{first}
\hl_1(x') \ge \hl_2(1 - x_1) \ge \ell_2(1 - x_1) \ge 1 \text{ for } x_1 < x' \le 1. 
\end{equation}

The above definition of $r^*$ is a core element of our
  proof. In Lemma~\ref{lem: domain_of_r*} we restrict the domain of
  $r^*$, as well as the range of the modified latencies for efficient mechanisms
  (those ones with low $\ePoA$), ending up with the lower bound
  provided in (\ref{lb1}). Then in Lemmas~\ref{lem: large x},~\ref{lem: small
    x} we focus on the properties that the equilibria of efficient
  mechanisms should satisfy. In Lemma~\ref{lem: large x}, we bound from
  above the amount of equilibrium flow that uses the first
  link if the mechanism is efficient, while in Lemma~\ref{lem: small
    x} we obtain a second lower bound on the $\ePoA$. Finally, in
  Lemma~\ref{lem: summary} we summarize
  the above properties ending up with the lower bound of~(\ref{lb2}).

\begin{lemma}\label{lem: domain_of_r*} If $\hl_1(x_1) \ge 1$ or $r^* = \infty$ or $r^* \le 1$, $\ePoA \ge 1.2$.
\end{lemma}
\begin{proof} If $\hl_1(x_1) \ge 1$ , we have $\CNM(1)
\ge 1$ and hence $\ePoA \ge \PoA(1) \ge 1.2$. 

If $r^* = \infty$, $\ePoA(\infty) \ge 1 + 1/R$. For $R \le 4$, this
is at least $1.25$. 

If $r^* \leq 1$, there is an equilibrium flow $(x,y)$ with $x > x_1$ and $r = x
+ y \leq 1$. Then $\hl_1(x) \ge 1$ by inequality (\ref{first}). Also $\hl_2(y)
\ge 1$. Thus 
$\CNM(r) \ge 1 \ge r$ and hence $\ePoA(r) \ge r/\Copt(r)$. For $r \le 1$, we have 
\[ \frac{r}{\Copt(r)} = \frac{r(1 + R)}{r^2 + Rr - R/4} \ge 1 + \frac{R/4}{r^2 +
Rr - R/4} \ge 1 + \frac{R/4}{1 + R - R/4} = \frac{4 + 4R}{4 + 3R} = \PoA(1) \ge 1.2.\]
 \end{proof}

In the light of the Lemma above, we proceed under the assumption $\hl_1(x_1) <
1$, and hence $x_1 < 1$, and $1 < r^* < \infty$. 
Then $(x_1,0)$ is an equilibrium flow, since $\hl_1(x_1) < 1 \le \hl_2(y')$
for $0 < y' \le x_1$. Thus 
\begin{equation}\label{lb1}
\ePoA(x_1) \ge \frac{x_1^2}{\Copt(x_1)}. 
\end{equation}
By definition of $r^*$, MN routes at most $x_1$ on the first
link for any rate $r < r^*$ and for any $\epsilon > 0$ 
there is an $r < r^* + \epsilon$ such that $(x,r - x)$ with $x > x_1$ is an
equilibrium flow for MN. 

For $r < r^*$, any equilibrium flow $(x,r-x)$ has $x \le x_1$. Thus, for
$x' \in (x,r] \supseteq (x_1,r]$, $\hl_1(x') \ge \hl_2(r-x) \ge \ell_2(r-x) \ge
\ell_2(r - x_1)$. Since this inequality holds for any $r < r^*$, we have 
\begin{equation}\label{second}
\hl_1(x') \ge \ell_2(r^* - x_1) \quad\text{for}\quad x' \in (x_1,r^*).
\end{equation}

\newcommand{\calF}{{F}}\newcommand{\calFe}{{\calF_\epsilon}}\newcommand{\calFez}{{\calF_{\epsilon_0}}}

For $\epsilon > 0$, let 
\[   \calF_\epsilon = \set{(x,y)}{\text{$(x,y)$ is an equilibrium flow with
$r^* \le x + y \le r^* + \epsilon$ and $x > x_1$}}.\]
Observe that $\calFe$ is non-empty by definition of $r^*$. 

\begin{lemma}\label{lem: large x}If for arbitrarily small $\epsilon > 0$, there is a $(x,y) \in \calFe$ with
$x \ge r^*$, then $\ePoA \ge \PoA(1) \ge 1.2$. 
\end{lemma}
\begin{proof} 
Let $r = x  + y$. Then $\ePoA \ge r^2 / \Copt(r)$. Since this
inequality holds for arbitrarily small $\epsilon$, $$\ePoA \ge
(r^*)^2/\Copt(r^*) \ge \PoA(1) = 1.2.$$ \end{proof}

We proceed under the assumption that there is an $\epsilon_0 > 0$ such that 
$\calF_{\epsilon_0}$ contains no pair $(x,y)$ with $x \ge r^*$. 

\begin{lemma}\label{lem: small x} 
If for arbitrarily small $\epsilon \in (0,\epsilon_0)$, $\calFe$ contains
either a pair $(x,r^*-x_1)$ or pairs $(x,y)$ and
$(u,v)$ with $y \not= v$, then $\ePoA \ge \ell_2(r^* - x_1)r^*/\Copt(r^*)$. \end{lemma}
\begin{proof} Assume first that $\calFe$ contains a pair $(x,r^*-x_1)$ and let
$r = x + r^* - x_1$. Then 
\[ \CNM(r) = \hl_1(x) x + \hl_2 (r^* - x_1) (r^* - x_1) \ge
\ell_2(r^* - x_1) r\]
since $\hl_1(x) \ge \ell_2(r^* - x_1)$ by (\ref{second}). 

Assume next that $\calFe$ contains pairs $(x,y)$ and $(u,v)$ with $y \not=
v$. Then $\hl_1(x) \ge  \ell_2(r^* - x_1)$ and $\hl_1(u) \ge  \ell_2(r^* -
x_1)$ by (\ref{second}). We may assume, $y > v$. Let $r = x+ y$. 
Since $(u,v)$ is an equilibrium $\hl_2(y') \ge
\hl_1(u)$ for $y' \in (v,u+v)$ and hence $\hl_2(y) \ge \hl_1(u)$ . Thus 
\[ \CNM(r) = \hl_1(x) x + \hl_2 (y) y  \ge
\ell_2(r^* - x_1) r.\]

We have now shown that $\CNM(r) \ge \ell_2(r^* - x_1) r$ for $r$'s greater than
$r^*$ and arbitrarily close to $r^*$. Thus $\ePoA \ge \ell_2(r^* - x_1)r^*/\Copt(r^*)$. 
 \end{proof}

We proceed under the assumption that there is an $\epsilon_0 > 0$ such that 
$\calF_{\epsilon_0}$ contains no pair $(x,y)$ with $x \ge r^*$, no pair $(x,r^*
- x_1)$ and no two pairs with distinct second coordinate. In other words, there
is a $y_0 < r^* - x_1$ such that all pairs in $\calF_{\epsilon_0}$ have second
coordinate equal to $y_0$. \bigskip

%[[ {\bf I proceed under the assumption that $\hl_1$ and $\hl_2$ are monotonic, i.e., $x <
%x'$ implies $\hl_1(x) \le \hl_1(x')$ and similarly for $\hl_2$. }\bigskip

Let $(x_0,y_0) \in \calF_{\epsilon_0}$. Then $y_0 < r^* - x_1$. Let $(x,y)$ be an equilibrium for
rate $r = (r^* + x_1 + y_0)/2$. Then $r = (2r^* + y_0 - (r^* - x_1))/2 < r^*$
and hence $x \le x_1$. Thus $y = r - x \ge r - x_1 = (r^* - x_1 + y_0)/2 >
y_0$ and $r - y_0 > x_1$. Consider the pair $(r - y_0,y_0)$. Its rate is less than $r^*$ and its
flow across the first link is $r - y_0$ which is larger than $x_1$. Thus it is
not an equilibrium by the definition of $r^*$. Therefore there is either an 
$x'' \in (r - y_0,r]$ with $\hl_1(x'') < \hl_2(y_0)$ or a $y'' \in (y_0,r]$ with
$\hl_2(y'') < \hl_1(r - y_0)$. We now distinguish cases. 

Assume the former. Since $(x,y)$ is an equilibrium, we have $\hl_1(x') \ge
\hl_2(y)$ for all $x' \in (x,x+y]$ and in particular for $x''$; observe that 
$r - y_0 \ge x$ since $r - x = y > y_0$. Thus 
$\hl_2(y_0) > \hl_2(y)$, a contradiction to the monotonicity of $\hl_2$. 

Assume the latter. Since $(x_0,y_0)$ is an equilibrium, we have $\hl_2(y') \ge
\hl_1(x_0)$ for all $y' \in (y_0,x_0 + y_0]$ and in particular for $y''$. Thus 
$\hl_1(r - y_0) > \hl_1(x_0)$, a contradiction to the monotonicity of $\hl_1$;
observe that $r - y_0 < x_0$ since $r < r^* \le x_0 + y_0$. 

\begin{lemma}\label{lem: summary}
\begin{equation}\label{lb2}
\ePoA \ge \min \left\{ 1.2, \min_{x_1 \le 1} \max\left\{ \frac{x_1^2}{\Copt(x_1)}, \min_{r^* \ge
1}
%\left\{
	\frac{\ell_2(r^* - x_1) r^*}{\Copt(r^*)} 
%	\right\} 
	\right\} \right\} .\end{equation}
\end{lemma}
\begin{proof} If $x_1 \ge 1$ or $r^* \le 1$ or $r^* = \infty$, we have $\ePoA
\ge 1.2$. So assume $x_1 < 1$ and $1 < r^* < \infty$. The argument preceding
this Lemma shows that the hypothesis of either Lemma~\ref{lem: large x}
or~\ref{lem: small x} is satisfied. In the former case, $\ePoA \ge 1.2$. In the
latter case, $\ePoA \ge \max \left\{\frac{x_1^2}{\Copt(x_1)}, \frac{\ell_2(r^* -
x_1) r^*}{\Copt(r^*)}\right\}$. This completes the proof.  \end{proof}

It remains to bound 
\begin{equation}\label{lb3}
\min_{x_1 \le 1} \max\left\{ \frac{x_1^2}{\Copt(x_1)}, \min_{r^* \ge
1} \frac{\ell_2(r^* - x_1) r^*}{\Copt(r^*)} \right\} = \min_{x_1 \le 1} \max
\left\{ \frac{x_1^2}{\Copt(x_1)}, \min_{r^* \ge
1} \frac{4 r^* (R+1) (r^* - x_1 + R)}{R(4(r^*)^2 + 4Rr^* - R)}\right\}\end{equation}
from below. We prove a lower bound of $1.191$. 
The term $x_1^2/\Copt(x_1)$ is increasing in $x_1$. Thus there is a unique
value $\alpha_1 \in [1/2,1]$ such that the first term is larger than $1.191$
for $x_1 > 
\alpha_1$. If the minimizing $x_1$ is larger than $\alpha_1$ we have
established the bound.

The second term is minimized for $r^* = \max\left\{1,\left(R + \sqrt{R^2 + 4 R^2 x_1 -
4 R x_1^2}\right)/(4x_1)\right\}$.  
Since $x_1 \le 1$ and hence $x_1^2 \le x_1$, we have
$(R + \sqrt{R^2 + 4 R^2 x_1 -
4 R x_1^2})/(4x_1)) \ge 2R/4 \ge 1$ and hence $$r^* = \left(R + \sqrt{R^2 + 4 R^2 x_1 -
4 R x_1^2}\right)/(4x_1).$$ The second term is decreasing in $x_1$ and hence we
may substitute $x_1$ by $\alpha_1$ for the purpose of establishing a lower
bound. We now specialize $R$ to $21/10$. For this
value of $R$ and $x_1 = \alpha_1$ 
\[ \frac{\ell_2(r^* - \alpha_1)
r^* }{\Copt(r^*)} |_{r^* = (R + \sqrt{R^2 + 4 R^2 \alpha_1 -
4 R \alpha_1^2})/(4\alpha_1) \text{ and } R = 21/10} \ge 1.191.\]
This completes the proof of the lower bound.\medskip

We next argue that the construction of Section~\ref{sec: two
links advanced} is optimal. Equations (\ref{eq: upper bound}) of Section~\ref{sec: two
links advanced} for $b_1 = 0$ and Equation (\ref{lb3}) agree. Hence our refined solution is optimal.
 \end{proof}

\section{Open Problems}
\label{sec:conclusion}

Clearly the ultimate goal is to design coordination mechanisms that work for general
networks. 
In the case of parallel links that we studied, we showed that our mechanism 
approaches $4/3$, as the number of links $k$ grows. 
It is still an open problem to show a bound of the form $4/5 -\alpha$, for some
strictly positive $\alpha$. 
%Can we improve the upper bound for the case of $k$ parallel links? 
A possible approach could be to use the ideas of 
Section~\ref{sec: two links advanced}. Another approach would be to define the
benign case more restrictively. Assuming $R_i = 8$ for all $i$, we would call
the following latencies benign: 
$\ell_1(x) = x$, and $\ell_i(x) = 1 + \epsilon\cdot i  + x/8^i$ for $i > 1$ and
small positive $\epsilon$. However, Opt starts using
the $k$-th link shortly after $1/2$ and hence uses an extremely efficient link
for small rates.  

Also, our results hold only for affine original latency functions.What can be
said for the case of more general latencies, for instance polynomials?
On the more technical side, it would be interesting to study whether our lower bound
construction of Section~\ref{lower bound} can be extended to modified latency 
functions $\hl$ that do not need to satisfy monotonicity.
%What can be said about atomic (weighted or unweighted)
%scheduling games and for games with polynomial latencies?
%What is the exact value of $\ePoA$ for the case
%of two parallel links? We conjecture that a reinspection of Sections~\ref{sec:
%two links advanced} and~\ref{lower bound} settles this question.

\paragraph{Acknowledgements}
We would like to thank Elias Koutsoupias, Spyros Angelopoulos and Nicol\'as Stier Moses
for many fruitful discussions.

%\bibliographystyle{alpha}
%\bibliographystyle{plain}
%\bibliography{coord,new}

\end{document}